\documentclass[romanappendices,draftcls,onecolumn,12pt]{IEEEtran}
\pdfoutput=1

\usepackage{epsfig}
\usepackage{amsmath}
\usepackage{amsfonts}
\usepackage{amssymb}
\usepackage{subfigure}
\usepackage{alecom}
\usepackage{booktabs}
\usepackage{citesort}
\usepackage{enumerate}
\usepackage[nolists,nomarkers]{endfloat}

\newtheorem{theorem}{Theorem}[section]
\newtheorem{proposition}{Proposition}[section]

\newtheorem{corollary}{Corollary}[section]

\newcommand{\mmse}{{\sf mmse}}

\newcommand{\snrbar}{\overline{\mathsf{snr}}}
\DeclareMathOperator{\K}{K}

\title{Outage Probability of the Gaussian MIMO Free-Space Optical Channel with PPM} 

\author{Nick Letzepis,\thanks{N. Letzepis is with Institute for Telecommunications Research, University of South Australia, SPRI Building - Mawson Lakes Blvd., Mawson Lakes SA 5095, Australia, e-mail: \tt nick.letzepis@unisa.edu.au.}~\IEEEmembership{Member,~IEEE} and Albert Guill\'en i F\`abregas,\thanks{A. Guill\'en i F\`abregas is with the Department of Engineering, University of Cambridge, Cambridge CB2 1PZ, UK, e-mail: {\tt guillen@ieee.org}.}~\IEEEmembership{Member,~IEEE}}

\date{\today}
\begin{document}

\maketitle

\begin{abstract}
The free-space optical channel has the potential to facilitate
inexpensive, wireless communication with fiber-like bandwidth under
short deployment timelines. However, atmospheric effects can
significantly degrade the reliability of a free-space optical link. In
particular, atmospheric turbulence causes random fluctuations in the
irradiance of the received laser beam, commonly referred to as
\textit{scintillation}. The scintillation process is slow compared to
the large data rates typical of optical transmission. As such, we
adopt a quasi-static block fading model and study the outage
probability of the channel under the assumption of orthogonal
pulse-position modulation. We investigate the mitigation of
scintillation through the use of multiple lasers and multiple
apertures, thereby creating a multiple-input multiple output (MIMO)
channel. Non-ideal photodetection is also assumed such that the
combined shot noise and thermal noise are considered as
signal-independent additive Gaussian white noise. Assuming perfect
receiver channel state information (CSI), we compute the
signal-to-noise ratio exponents for the cases when the scintillation
is lognormal, exponential and gamma-gamma distributed, which cover a
wide range of atmospheric turbulence conditions. Furthermore, we
illustrate very large gains, in some cases larger than $15$ dB, when
transmitter CSI is also available by adapting the transmitted
electrical power.
\end{abstract}
 
 
\section{Introduction}
\PARstart{F}{ree}-space optical (FSO) communication offers an
attractive alternative to the radio frequency (RF) channel for the
purpose of transmitting data at very high rates. By utilising a high
carrier frequency in the optical range, digital communication on the
order of gigabits per second is possible. In addition, FSO links are
difficult to intercept, immune to interference or jamming from
external sources, and are not subject to frequency spectrum
regulations. FSO communications have received recent attention in
applications such as satellite communications, fiber-backup,
RF-wireless back-haul and last-mile connectivity~\cite{wille02}.

The main drawback of the FSO channel is the
detrimental effect the atmosphere has on a propagating laser beam.
The atmosphere is composed of gas molecules, water vapor, pollutants,
dust, and other chemical particulates that are trapped by Earth's
gravitational field. Since the wavelength of a typical optical carrier
is comparable to these molecule and particle sizes, the carrier wave
is subject to various propagation effects that are uncommon to RF
systems. One such effect is \textit{scintillation}, caused by
atmospheric turbulence, and refers to random fluctuations in the
irradiance of the received optical laser beam (analogous to fading
in RF systems)~\cite{stro78,gali95,andr05}.

Recent works on the mitigation of scintillation concentrate on the use
of multiple-lasers and multiple-apertures to create a
multiple-input-multiple-output (MIMO)
channel~\cite{LetzepisHollandCowley2008,ChakDeyFran07,ChakPrak07,cvi07,djor06,wilson05aug,chak05,lee04,haas03}. Many
of these works consider scintillation as an ergodic fading process,
and analyse the channel in terms of its ergodic capacity. However,
compared to typical data rates, scintillation is a slow time-varying
process (with a coherence time on the order of milliseconds), and it
is therefore more appropriate to analyse the outage probability of the
channel. To some extent, this has been done in the works
of~\cite{ChakDeyFran07,far07,wilson05aug,lee04,haas03}. In~\cite{ChakDeyFran07,haas03}
the outage probability of the MIMO FSO channel is analysed under the
assumption of ideal photodetection (i.e. a Poisson counting process)
with no bandwidth constraints. Wilson \textit{et
al.}~\cite{wilson05aug} also assume perfect photodetection, but with
the further constraint of pulse-position modulation (PPM).  Lee and
Chan~\cite{lee04}, study the outage probability under the assumption
of on-off keying (OOK) transmission and non-ideal photodetection,
i.e. the combined shot noise and thermal noise process is modeled as
zero mean signal independent additive white Gaussian noise
(AWGN). Farid and Hranilovic~\cite{far07} extend this analysis to
include the effects of pointing errors.

\begin{figure*}[ht]
  \centering 
  \includegraphics[width=0.8\columnwidth]{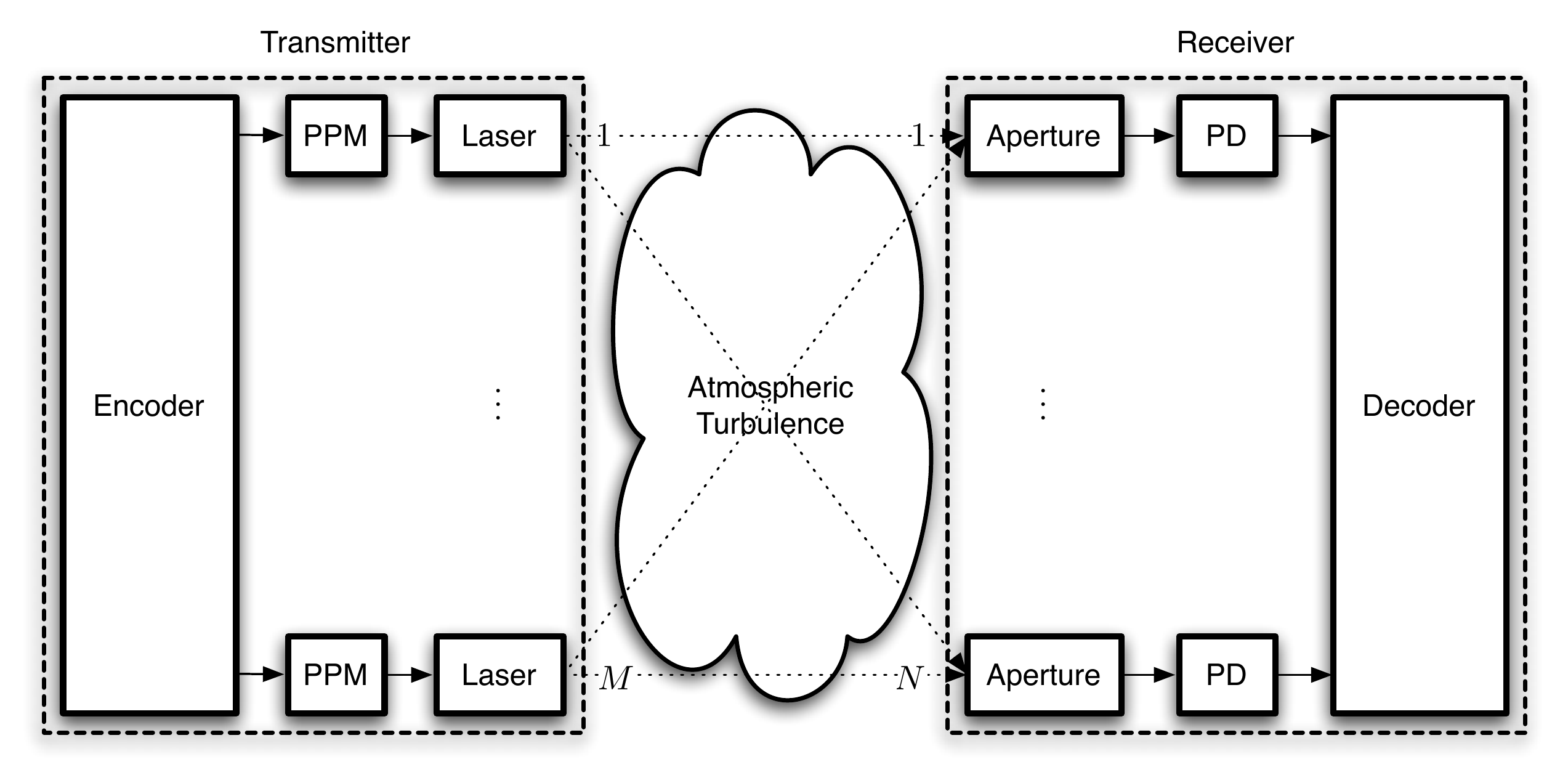}  
  \caption{Block diagram of an $M\times N$ MIMO FSO system.}
  \label{fig:mimo_model}
\end{figure*}

In this paper we study the outage probability of the MIMO FSO channel
under the assumptions of PPM, non-ideal photodetection, and equal gain
combining (EGC) at the receiver. In particular, we model the channel
as a quasi-static block fading channel whereby communication takes
place over a finite number of blocks and each block of transmitted
symbols experiences an independent identically distributed (i.i.d.)
fading realisation
\cite{ozarow_shamai_wyner,biglieri_proakis_shamai}. We consider two
types of CSI knowledge. First we assume perfect CSI is available only
at the receiver (CSIR case), and the transmitter knows only the
channel statistics. Then we consider the case when perfect CSI is also
known at the transmitter (CSIT case).\footnote{Given the slow
time-varying scintillation process, CSI can be estimated at the
receiver and fed back to the transmitter via a dedicated feedback
link.} Under this framework we study a number of scintillation
distributions: lognormal, modelling weak turbulence; exponential,
modelling strong turbulence; and gamma-gamma~\cite{haba01}, which
models a wide range of turbulence conditions.  For the CSIR case, we
derive signal-to-noise ratio (SNR) exponents and show that they are
the product of: a channel related parameter, dependent on the
scintillation distribution; the number of lasers times the number of
apertures, reflecting the spatial diversity; and the Singleton bound
\cite{KnoppHumblet,malkamaki_leib,albert_beppe_it}, reflecting the
block diversity. For the CSIT case, the transmitter finds the optimal
power allocation that minimises the outage probability
\cite{caire_taricco_biglieri_pc}. Using results
from~\cite{NguyenGuillenRasmussen2007}, we derive the optimal power
allocation subject to short- and long-term power constraints. We show
that very large power savings are possible compared to the CSIR
case. Interestingly, under a long-term power constraint, we show that
delay-limited capacity~\cite{HanlyTse1998} is zero for exponential and
(in some cases) gamma-gamma scintillation, unless one codes over
multiple blocks, and/or uses multiple lasers and apertures.

The paper is organised as follows. In Section~\ref{sec:system_model},
we define the channel model and assumptions. In
Section~\ref{sec:scint_models} we review the lognormal, exponential
and gamma-gamma models. Section~\ref{sec:pout_mut_mmse} defines the
outage probability and presents results on the minimum-mean squared
error (MMSE). Then in Sections~\ref{sec:pout_anal_csir} and
\ref{sec:pout_anal_csit} we present the main results of our asymptotic
outage probability analysis for the CSIR and CSIT cases,
respectively. Concluding remarks are then given in
Section~\ref{sec:conc}. Proofs of the various results can be found in
the Appendices.

\section{System Model} \label{sec:system_model}
We consider an $M \times N$ MIMO FSO system with $M$ transmit lasers
an $N$ aperture receiver as shown in
Fig.~\ref{fig:mimo_model}. Information data is first encoded by a
binary code of rate $R_c$. The encoded stream is modulated according
to a $Q$-ary PPM scheme, resulting in rate $R=R_c\log_2Q$
(bits/channel use). Repetition transmission is employed such that the
same PPM signal is transmitted in perfect synchronism by each of the
$M$ lasers through an atmospheric turbulent channel and collected by
$N$ receive apertures. We assume the distance between the individual
lasers and apertures is sufficient so that spatial correlation is
negligible. At each aperture, the received optical signal is converted
to an electrical signal via photodetection. Non-ideal photodetection is assumed such that the
combined shot noise and thermal noise processes can be modeled as zero
mean, signal independent AWGN (an assumption commonly used in the
literature, see
e.g.~\cite{gali95,doli00b,doli00,zhu03,li03,lee04,simon05,andr05,far07,navi07,LetzepisHollandCowley2008}).

In FSO communications, channel variations are typically much slower
than the signaling period. As such, we model the channel as a
non-ergodic block-fading channel, for which a given codeword of length
$BL$ undergoes only a finite number $B$ of scintillation realisations
\cite{ozarow_shamai_wyner,biglieri_proakis_shamai}. The received signal at aperture $n$, $n = 1,\ldots, N$ can be
written as
\begin{equation} \label{eq:mimo_model}
\yvv_b^n [\ell] =  \left( \sum_{m=1}^{M} \tilde{h}_b^{m,n} \right)\sqrt{\tilde{p}_b}  \,\xvv_b[\ell] + \tilde{\zvv}^n_b[\ell],
\end{equation}
for $b = 1,\dotsc,B, \ell = 1,\dotsc,L$, where $\yvv_{b}^n[\ell],
\tilde{\zvv}_{b}^n[\ell]\in \RR^Q$ are the received and noise signals
at block $b$, time instant $\ell$ and aperture $n$, $\xvv_{b}[\ell],
\in \RR^Q$ is the transmitted signal at block $b$ and time instant
$\ell$, and $\tilde{h}_b^{m,n}$ denotes the scintillation fading
coefficient between laser $m$ and aperture $n$. Each transmitted
symbol is drawn from a PPM alphabet, $\xvv_b[\ell] \in \Xc^{\rm ppm}
\eqdef \{\evv_1,\dotsc,\evv_Q\}$, where $\evv_q$ is the canonical
basis vector, i.e., it has all zeros except for a one in position $q$,
the time slot where the pulse is transmitted. The noise samples of
$\tilde{\zvv}^n_{b}[\ell]$ are independent realisations of a random
variable $Z \sim \Nc(0,1)$, and $\tilde{p}_b$ denotes the received
electrical power of block $b$ at each aperture in the absence of
scintillation. The fading coefficients $\tilde{h}_b^{m,n}$ are
independent realisations of a random variable $\tilde{H}$ with
probability density function (pdf) $f_{\tilde{H}}(h)$.

At the receiver, we assume equal gain combining (EGC) is employed,
such that the entire system is equivalent to a single-input
single-output (SISO) channel, i.e.
\begin{align} 
\yvv_b[\ell] = \frac{1}{\sqrt{N}}\sum_{n=1}^{N} \yvv_b^n[\ell]
             = \sqrt{p_b}  h_b \xvv_b[\ell] + \zvv_b[\ell], \label{eq:mimo_egc_model2}
\end{align}
where $\zvv_b[\ell] = \frac{1}{\sqrt{N}} \sum_{n=1}^{N}
\tilde{\zvv}_b^n[\ell] \sim \mathcal{N}(0,1)$, and $h_b$, a
realisation of the random variable $H$, is defined as the normalised
combined fading coefficient, i.e. 
\beq h_b = \frac{c}{MN}
\sum_{m=1}^{M} \sum_{n=1}^{N} \tilde{h}_b^{m,n},
\eeq 
where $c = 1/(\EE[\tilde{H}] \sqrt{1 + \sigma^2_I/(MN)})$ is a constant to
ensure $\EE[H^2] = 1$.\footnote{For optical channels with ideal photodetection, the normalisation
$\EE[H] =1 $ is commonly used to keep optical power constant. We assume
non-ideal photodetection and work entirely in the electrical domain. Hence, we
chose the normalisation $\EE[H^2] = 1$, used commonly in RF fading
channels. However, since we consider only the asymptotic behaviour of the
outage probability, the specific normalisation is irrelevant and does
not affect our results.}%
~Thus, the total instantaneous received electrical power at block $b$
is $p_b = M^2 N \tilde{p}_b/c$, and the total average received SNR is
$\snr \triangleq \EE [h_b p_b] = \EE[p_b]$.

For the CSIR case, we assume the electrical power is distributed
uniformly over the blocks, i.e., $p_b=p=\snr$ for
$b=1,\dotsc,B$. Otherwise, for the CSIT case, we will allocate
electrical power in order to improve performance. In particular, we
will consider the following two electrical power constraints
\begin{align}
\text{Short-term:} ~~~~& \frac{1}{B}\sum_{b=1}^B p_b\leq P \label{eq:st_pc}\\
\text{Long-term:} ~~~~& \EE\left[\frac{1}{B}\sum_{b=1}^B p_b\right]\leq P\label{eq:lt_pc}.
\end{align}

Throughout the paper, we will devote special attention to the case of
$B=1$, i.e., the channel does not vary within a codeword. This
scenario is relevant for FSO, since, due to the large data-rates, one
is able to transmit millions of bits over the same channel
realisation. We will see that most results admit very simple forms,
and some cases, even closed form. This analysis allows for a system characterisation where the expressions highlight the roles of the key design parameters.

\section{Scintillation Distributions}\label{sec:scint_models}
The scintillation pdf, $f_{\tilde{H}}(h)$, is parameterised by the
{\em scintillation index} (SI), 
\beq \label{eq:si}
\sigma_I^2 \triangleq \frac{\mathrm{Var}(\tilde{H})}{(\EE[\tilde{H}])^2},
\eeq 
and can be considered as a measure of the strength of the optical
turbulence under weak turbulence conditions~\cite{andr99,haba01}.

The distribution of the irradiance fluctuations is dependent on the
strength of the optical turbulence. For the weak turbulence regime,
the fluctuations are generally considered to be lognormal
distributed, and for very strong turbulence, exponential
distributed~\cite{lawr70,stro78}. For moderate turbulence, the
distribution of the fluctuations is not well understood, and a number
of distributions have been proposed, such as the lognormal-Rice
distribution~\cite{chur87,hill97,haba01,andr05,Vet07} (also known as
the Beckmann distribution~\cite{beck67}) and
K-distribution~\cite{chur87}. In~\cite{haba01},
Al-Habash~\textit{et~al.} proposed a gamma-gamma distribution as a
general model for all levels of atmospheric turbulence. Moreover,
recent work in~\cite{Vet07} has shown that the gamma-gamma model is in
close agreement with experimental measurements under
moderate-to-strong turbulence conditions. In this paper we focus on
lognormal, exponential, and gamma-gamma distributed
scintillation, which are described as follows.

For lognormal distributed scintillation,
\beq
\label{eq:lognormal_pdf} 
f^{\ln}_{\tilde{H}}(h) = \frac{1}{h \sigma \sqrt{2 \pi}}
\exp \left(-(\log h - \mu)^2/(2\sigma^2) \right),
\eeq 
where $\mu$ and $\sigma$ are related to the SI via $\mu =
-\log(1+\sigma_I^2)$ and $\sigma^2 = \log(1+\sigma_I^2)$.

For exponential distributed scintillation
\beq \label{eq:exp_scint}
f^{\exp}_{\tilde{H}}(h) = \lambda \exp(- \lambda h)
\eeq
which corresponds to the super-saturated turbulence regime, where $\sigma^2_I = 1$.

The gamma-gamma distribution arises from the product of two
independent Gamma distributed random variables and has the
pdf~\cite{haba01},
\begin{equation} \label{eq:gg_dist}
f^{\rm gg}_{\tilde{H}}(h) = \frac{2 (\alpha \beta)^{\frac{\alpha+\beta}{2}}}{\Gamma(\alpha) \Gamma(\beta)} h^{\frac{\alpha +\beta}{2} - 1} \K_{\alpha - \beta} (2 \sqrt{\alpha \beta h}),
\end{equation}
where $\K_{\nu}(x)$ denotes the modified Bessel function of the second
kind~\cite[Ch.~10]{abramowitz_stegun}. The parameters $\alpha$ and
$\beta$ are related with the scintillation index via $\sigma_I^2=
\alpha^{-1} + \beta^{-1} + (\alpha\beta)^{-1}$.

\section{Information Theoretic Preliminaries} \label{sec:pout_mut_mmse}
The channel described by \eqref{eq:mimo_egc_model2} under the
quasi-static assumption is not information stable \cite{verdu_han} and
therefore, the channel capacity in the strict Shannon sense is
zero. It can be shown that the codeword error probability of any
coding scheme is lower bounded by the information outage probability
\cite{ozarow_shamai_wyner,biglieri_proakis_shamai}, \beq
\label{eq:pout_def} P_{\rm out}(\snr,R) = \Pr(I(\pvv,\hvv)<R), \eeq
where $R$ is the transmission rate and $I(\pvv,\hvv)$ is the
instantaneous input-output mutual information for a given power
allocation $\pvv \triangleq (p_1, \ldots, p_B)$, and vector channel
realisation $\hvv \triangleq (h_1, \ldots, h_B)$. The instantaneous
mutual information can be expressed as \cite{cover_thomas}
 \beq \label{eq:instant_mi}
I(\pvv,\hvv) = \frac{1}{B}\sum_{b=1}^B I^{\rm awgn}(p_b h_b^2),
\eeq
where $I^{\rm awgn}(\rho)$ is the input-output mutual information of
an AWGN channel with SNR $\rho$. For PPM \cite{doli00b}
\beq
I^{\rm awgn}(\rho)  = \log_2 Q - \EE \left[\log_2\left( 1 + \sum_{q=2}^Q e^{-\rho + \sqrt{\rho}(Z_q-Z_1) } \right)\right],
\label{eq:mi_ppm}
\eeq
where $Z_q \sim \Nc(0,1)$ for $q = 1,\ldots, Q$.

For the CSIT case we will use the recently discovered relationship
between mutual information and the MMSE
\cite{GuoShamaiVerdu2005}. This relationship states that\footnote{The
$\log(2)$ term arises because we have defined $ I^{\rm awgn}(\rho)$ in
bits/channel usage.}
\beq
\frac{d}{d\rho} I^{\rm awgn}(\rho)  = \frac{\mmse(\rho)}{\log(2)}
\eeq
where $\mmse(\rho)$ is the MMSE in estimating the input from the
output of a Gaussian channel as a function of the SNR $\rho$. For PPM,
we have the following result
\begin{theorem} \label{theorem:mmse}
The MMSE for PPM on the AWGN channel with SNR $\rho$ is
\begin{align}
\mmse (\rho) = 1 - \EE \left[ \frac{ e^{2 \sqrt{\rho} (\sqrt{\rho} + Z_1)} + (Q-1)e^{2 \sqrt{\rho} Z_2}}{ \left( e^{\rho + \sqrt{\rho} Z_1 } + \sum_{k=2}^{Q} e^{\sqrt{\rho} Z_k} \right)^2} \right],  \label{eq:qppm_mmse}
\end{align}
where $Z_i \sim \Nc (0,1)$ for $i = 1,\ldots, Q$.
\end{theorem}
\begin{proof}
See Appendix~\ref{appendix:proof_mmse}.
\end{proof}

Note that both~\eqref{eq:mi_ppm} and~\eqref{eq:qppm_mmse} can be
evaluated using standard Monte-Carlo methods.

\section{Outage Probability Analysis with CSIR} \label{sec:pout_anal_csir}

For the CSIR case, we employ uniform power allocation, i.e. $p_1 =
\ldots = p_B = \snr$. For codewords transmitted over $B$ blocks,
obtaining a closed form analytic expression for the outage probability
is intractable. Even for $B=1$, in some cases, for example the
lognormal and gamma-gamma distributions, determining the exact
distribution of $H$ can be a difficult task.  Instead, as we shall
see, obtaining the asymptotic behaviour of the outage probability is
substantially simpler. Towards this end, and following the footsteps
of \cite{Tse,albert_beppe_it}, we derive the~\textit{SNR~exponent}.

\begin{theorem} \label{theorem:mimo_snr_exponents}
The outage SNR exponents for a MIMO FSO communications system modeled
by~\eqref{eq:mimo_egc_model2} are given as follows:
\renewcommand{\labelenumi}{\textit{(\alph{enumi})}}
\begin{align}
 d^{\rm ln}_{(\log \snr)^2} &= \frac{ MN }{8\log(1+\sigma_I^2)} \left( 1 + \left\lfloor B\left(1- R_c\right) \right\rfloor \right)\label{eq:mimo_snr_exponent_lognormal}\\
 d^{\rm exp}_{(\log \snr)} &= \frac{MN}{2}  \left( 1 + \left\lfloor B\left(1- R_c\right) \right\rfloor \right),  \label{eq:mimo_snr_exponent_exponential}\\
 d^{\rm gg}_{(\log\snr)} &= \frac{MN}{2} \min(\alpha,\beta) \left(1 + \left\lfloor B\left(1- R_c\right) \right\rfloor \right), \label{eq:mimo_snr_exponent_gamma_gamma}
\end{align}
for lognormal, exponential, and gamma-gamma cases respectively, 
where $R_c = R/\log_2(Q)$ is the rate of the binary code and
\begin{align}
d_{(\log\snr)^k} &\eqdef -\lim_{\snr\to\infty}\frac{\log P_{\rm out}(\snr,R)}{(\log\snr)^k}~~~ k=1,2.
\end{align}
\end{theorem}
\begin{proof}
See Appendix~\ref{appendix:proof_csir_snr_exponents}.
\end{proof}

\begin{proposition} \label{theorem:mimo_snr_exponents_achievability}
The outage SNR exponents given in Theorem
\ref{theorem:mimo_snr_exponents}, are achievable by random coding over
PPM constellations whenever $ B\left(1- R_c\right)$ is not an integer.
\end{proposition}
\begin{proof}
The proof follows from the proof of Theorem \ref{theorem:mimo_snr_exponents} and the proof of \cite[Th. 1]{albert_beppe_it}.
\end{proof}

The above proposition implies that the outage exponents given in
Theorem~\ref{theorem:mimo_snr_exponents} are the optimal SNR exponents
over the channel, i.e. the outage probability is a lower bound to the
error probability of any coding scheme, its corresponding exponents
(given in Theorem \ref{theorem:mimo_snr_exponents}) are an upper bound
to the exponent of coding schemes. From Proposition
\ref{theorem:mimo_snr_exponents_achievability}, we can achieve the
outage exponents with a particular coding scheme (random coding, in
this case), and therefore, the exponents in
Theorem~\ref{theorem:mimo_snr_exponents} are optimal.

From~\eqref{eq:mimo_snr_exponent_lognormal}-\eqref{eq:mimo_snr_exponent_gamma_gamma}
we immediately see the benefits of spatial and block diversity on the
system. In particular, each exponent is proportional to: the number of
lasers times the number of apertures, reflecting the spatial
diversity; a channel related parameter that is dependent on the
scintillation distribution; and the Singleton bound, which is the
optimal rate-diversity tradeoff for Rayleigh-faded block fading
channels~\cite{KnoppHumblet,malkamaki_leib,albert_beppe_it}.

Comparing the channel related parameters
in~\eqref{eq:mimo_snr_exponent_lognormal}-\eqref{eq:mimo_snr_exponent_gamma_gamma}
the effects of the scintillation distribution on the outage
probability are directly visible. For the lognormal case, the channel
related parameter is $8 \log (1 + \sigma^2_I)$ and hence is directly
linked to the SI. Moreover, for small $\sigma^2_I < 1$, $8 \log (1 +
\sigma^2_I) \approx 8 \sigma^2_I$ and the SNR exponent is inversely
proportional to the SI. For the exponential case, the channel related
parameter is a constant $1/2$ as expected, since the SI is
constant. For the gamma-gamma case the channel related parameter is
$\min(\alpha, \beta)/2$, which highlights an interesting connection
between the outage probability and recent results in the theory of
optical scintillation. For gamma-gamma distributed scintillation, the
fading coefficient results from the product of two independent random
variables, i.e. $\tilde{H} = XY$, where $X$ and $Y$ model fluctuations
due to large scale and small scale cells. Large scale cells cause
refractive effects that mainly distort the wave front of the
propagating beam, and tend to steer the beam in a slightly different
direction (i.e. beam wander). Small scale cells cause scattering by
diffraction and therefore distort the amplitude of the wave through
beam spreading and irradiance fluctuations~\cite[p.~160]{andr05}. The
parameters $\alpha,\beta$ are related to the large and small scale
fluctuation variances via $\alpha = \sigma_X^{-2}$ and $\beta =
\sigma_Y^{-2}$. For a plane wave (neglecting inner/outer scale
effects) $\sigma_Y^2 > \sigma^2_X$, and as the strength of the optical
turbulence increases, the small scale fluctuations dominate and
$\sigma^2_Y \rightarrow 1$~\cite[p.~336]{andr05}. This implies that
the SNR exponent is exclusively dependent on the small scale
fluctuations. Moreover, in the strong fluctuation regime, $\sigma^2_Y
\rightarrow 1$, the gamma-gamma distribution reduces to a
K-distribution~\cite[p.~368]{andr05}, and the system has the same SNR
exponent as the exponential case typically used to model very strong
fluctuation regimes.

In comparing the lognormal exponent with the other cases, we observe a
striking difference. For the lognormal
case~\eqref{eq:mimo_snr_exponent_lognormal} implies the outage
probability is dominated by a $(\log(\snr))^2$ term, whereas for
exponential and gamma-gamma scintillation it is dominated by a $\log
(\snr)$ term. Thus the outage probability decays much more rapidly
with SNR for the lognormal case than it does for the exponential or
gamma-gamma cases. Furthermore, for the lognormal case, the slope of
the outage probability curve, when plotted on a $\log$-$\log$ scale,
will not converge to a constant value. In fact, a constant slope curve
will only be observed when plotting the outage probability on a
$\log$-$(\log)^2$ scale.

In deriving~\eqref{eq:mimo_snr_exponent_lognormal} (see
Appendix~\ref{appendix:proof_csir_snr_exponents:lognormal}) we do not
rely on the lognormal approximation\footnote{This refers to
approximating the distribution of the sum of lognormal distributed
random variables as
lognormal~\cite{mitchell68,slimane01,beau04a,beau04b}.}, which has
been used on a number occasions in the analysis of FSO MIMO channels,
e.g.~\cite{lee04,navi07,LetzepisHollandCowley2008}. Under this
approximation, $H$ is lognormal distributed~\eqref{eq:lognormal_pdf}
with parameters $\mu = -\log(1 + \sigma^2_I/(MN)) $ and $\sigma^2 =
-\mu$, and we obtain the approximated exponent
 \beq d_{(\log\snr)^2}
\approx \frac{1}{8\log(1+\frac{\sigma_I^2}{MN})} \left ( 1 +
\left\lfloor B\left(1-R_c\right) \right\rfloor \right).
\label{eq:exponent_mimo_lognormal_approx}
\eeq
Comparing~\eqref{eq:mimo_snr_exponent_lognormal}
and~\eqref{eq:exponent_mimo_lognormal_approx} we see that although the
lognormal approximation also exhibits a $(\log(\snr))^2$ term, it has
a different slope than the true SNR exponent. The difference is due to
the approximated and true pdfs having different behaviours in the
limit as $h \rightarrow 0$. However, for very small $\sigma^2_I < 1$,
using $\log(1 + x) \approx x$ (for $x < 1$)
in~\eqref{eq:mimo_snr_exponent_lognormal}
and~\eqref{eq:exponent_mimo_lognormal_approx} we see that they are
approximately equal.

For the special case of single block transmission, $B=1$, it is
straightforward to express the outage probability in terms of the
cumulative distribution function (cdf) of the scintillation random
variable, i.e. 
\beq
P_{\rm out}(\snr,R) = F_{H} \left(\sqrt{ \frac{\snr_{R}^{\rm
awgn}}{\snr}} \right)
\eeq
where $F_H(h)$ denotes the cdf of $H$, and
$\snr_{R}^{\rm awgn} \eqdef I^{\rm awgn, -1}(R)$ denotes the SNR value
at which the mutual information is equal to $R$. Table
\ref{table:min_snr} reports these values for $Q=2,4,8,16$ and
$R=R_c\log_2 Q$, with $R_c =
\frac{1}{4},\frac{1}{2},\frac{3}{4}$. 
Therefore, for $B =
1$, we can compute the outage probability analytically when the
distribution of $H$ is available, i.e., in the exponential case for
$M,N\geq 1$ or in the lognormal and gamma-gamma cases for $M,N=1$. 
In the case of exponential scintillation we have that
\beq 
P_{\rm{out}} (\snr, R) =   \bar\Gamma \left( MN, \left( MN(1+MN) \frac{\snr^{\rm awgn}_{R}}{\snr} \right)^{\frac{1}{2}} \right) , \label{eq:pout_exp_mimo}
\eeq
where $\bar\Gamma(a,x) \triangleq \frac{1}{\Gamma(a)}\int_0^x t^{a-1}
\exp(-t) \, dt$ denotes the regularised (lower) incomplete gamma
function~\cite[p.260]{abramowitz_stegun}. For the lognormal and
gamma-gamma scintillation with $MN > 1$, we must resort to numerical
methods. This involved applying the fast Fourier transform (FFT) to
$f_{\tilde{H}}$ to numerically compute its characteristic function,
taking it to the $MN$th power, and then applying the inverse FFT to
obtain $f_{H}$. This method yields very accurate numerical
computations of the outage probability in only a few seconds.

\begin{table}[t]
\caption{Minimum signal-to-noise ratio $\snr_{R}^{\rm awgn}$ (in decibels) for reliable communication for target rate $R=R_c \log_2 Q$.} 
\centering
\begin{tabular}{@{}cccc@{}}
\toprule
	$Q$		& $R_c = \frac{1}{4}$	& $R_c = \frac{1}{2}$& $R_c = \frac{3}{4}$\\ \midrule
	$2$        	& $-0.7992$	    		& $3.1821$   		& $6.4109$      \\
	$4$		& $0.2169$			& $4.0598$		& $7.0773$      \\
	$8$		& $1.1579$    			& $4.8382$   		& $7.7222$	\\
	$16$		& $1.9881$    			& $5.5401$   		& $8.3107$      \\
\bottomrule
\end{tabular}
\label{table:min_snr}
\end{table}

Outage probability curves for the $B=1$ case are shown on the left in
Fig.~\ref{fig:outage_plot}.  For the lognormal case, we see that the
curves do not have constant slope for large SNR, while, for the
exponential and gamma-gamma cases, a constant slope is clearly
visible. We also see the benefits of MIMO, particularly in the
exponential and gamma-gamma cases, where the SNR
exponent has increased from $1/2$ and $1$ to $2$ and $4$
respectively.

\begin{figure*}[t]
  \centering\small

   \subfigure{\includegraphics*[width=0.49\columnwidth]{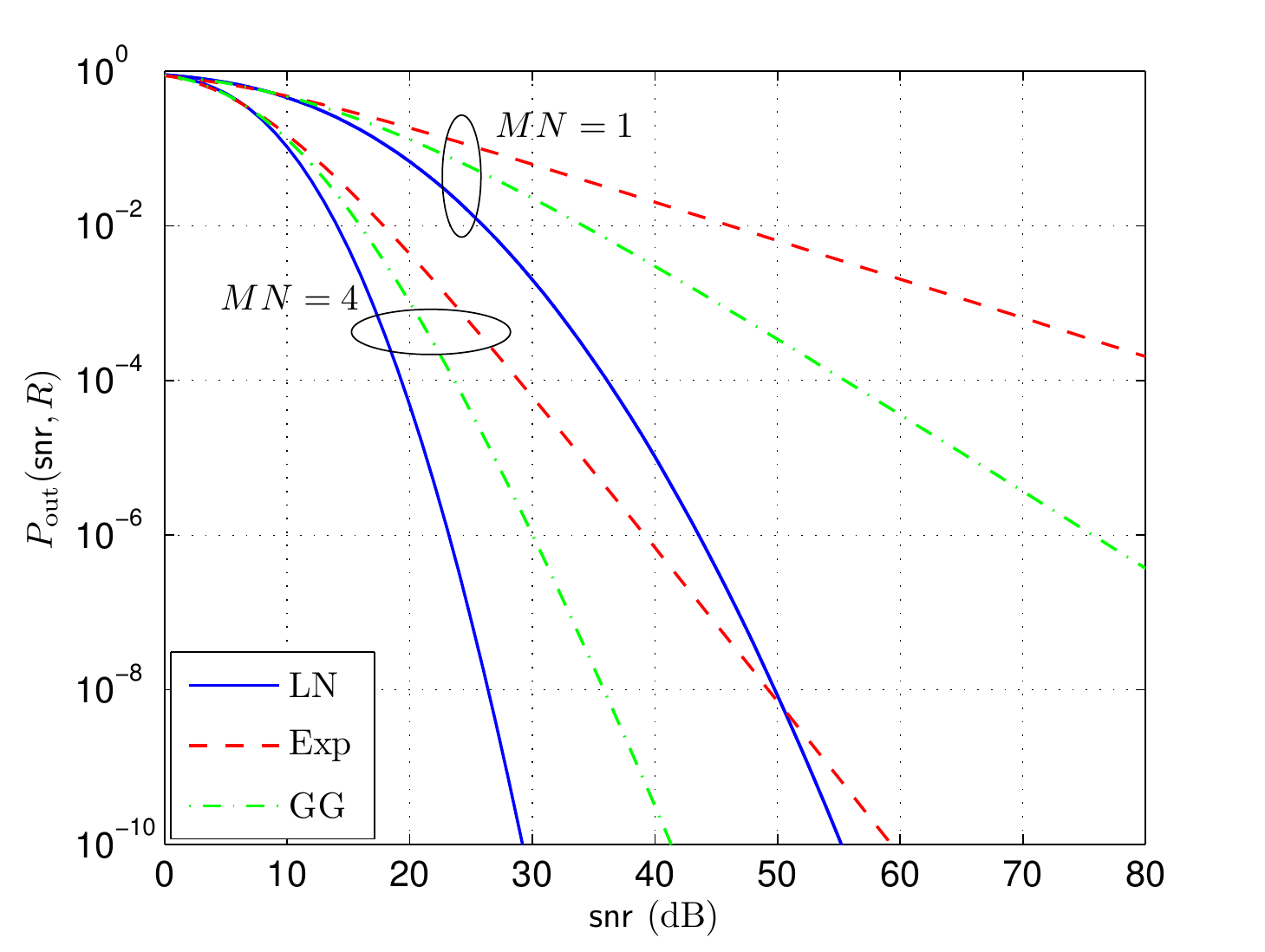}}
   \subfigure{\includegraphics*[width=0.49\columnwidth]{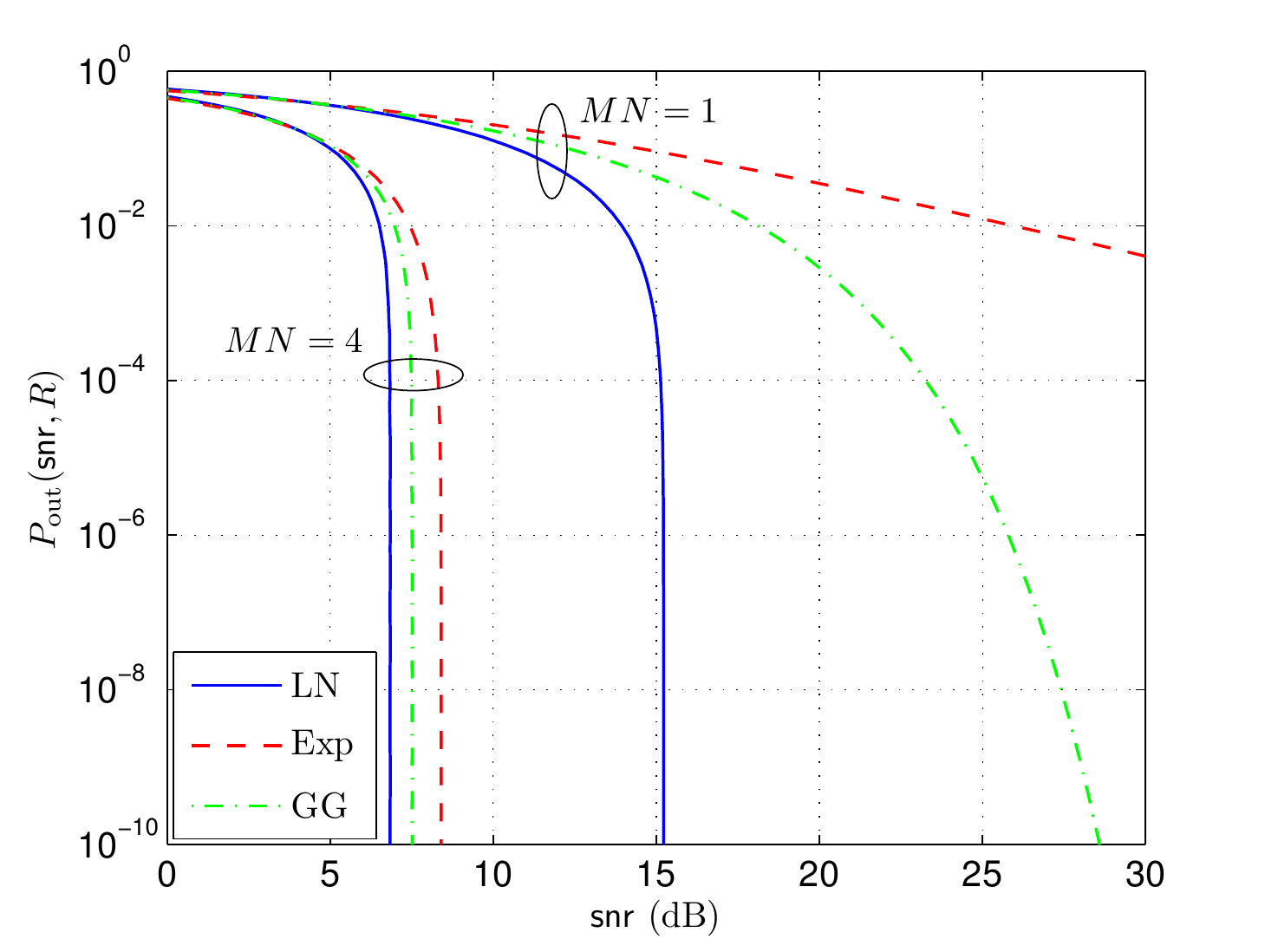}}
   
  \caption{Outage probability curves for the CSIR (left) and CSIT (right) cases with $\sigma^2_I = 1$,
  $B=1$, $Q=2$, $R_c=1/2$, $\snr^{\rm awgn}_{1/2} = 3.18$ dB:
  lognormal (solid); exponential (dashed); and,
  gamma-gamma distributed scintillation (dot-dashed), $\alpha = 2$,
  $\beta = 3$.}
  \label{fig:outage_plot}
\end{figure*}

\section{Outage Probability Analysis with CSIT}  \label{sec:pout_anal_csit}
In this section we consider the case where the transmitter and
receiver both have perfect CSI knowledge. In this case, the
transmitter determines the optimal power allocation that minimises the
outage probability for a fixed rate, subject to a power constraint
\cite{caire_taricco_biglieri_pc}. The results of this section are
based on the application of results
from~\cite{NguyenGuillenRasmussen2007} to PPM and the scintillation
distributions of interest. Using these results we uncover new insight
as to how key design parameters influence the performance of the
system. Moreover, we show that large power savings are possible
compared to the CSIR case.

For the short-term power constraint given by \eqref{eq:st_pc}, the optimal power allocation is given by
mercury-waterfilling at each channel realisation
\cite{LozanoTulinoVerdu2006,NguyenGuillenRasmussen2007},
\begin{equation}
p_b = \frac{1}{h_b^2}\mmse^{-1}\left(\min\left\{ \frac{Q-1}{Q}, \frac{\eta}{h_b^2}\right\}\right), \label{eq:st_power_alloc}
\end{equation}
for $b=1, \ldots, B$ where $\mmse^{-1}(u)$ is the inverse-MMSE
function and $\eta$ is chosen to satisfy the power
constraint.\footnote{Note that
in~\cite{LozanoTulinoVerdu2006,NguyenGuillenRasmussen2007}, the
minimum in~\eqref{eq:st_power_alloc} is between $1$ and
$\frac{\eta}{h_b^2}$. For $Q$PPM, $\mmse(0) = \frac{Q-1}{Q}$
(see~\eqref{eq:qppm_mmse}). Hence we must replace $1$ with $\frac{Q-1}{Q}$.}
From~\cite[Prop. 1]{NguyenGuillenRasmussen2007} it is apparent that
the SNR exponent for the CSIT case under short-term power constraints
is the same as the CSIR case.

For the long-term power constraint given by \eqref{eq:lt_pc} the optimal power
allocation is~\cite{NguyenGuillenRasmussen2007}
\begin{equation}
\label{eq:long_sol}
\pvv= \begin{cases}
\boldsymbol{\wp}, &\sum_{b=1}^B \wp_b \leq s\\
\mathbf{0}, &{\rm otherwise},
\end{cases}
\end{equation}
where
\begin{equation}
\label{eq:minpow_sol}
\wp_b  = \frac{1}{h^2_b}\mmse^{-1}\left(\min\left\{ \frac{Q-1}{Q}, \frac{1}{\eta h^2_b}\right\}\right), \;b=1, \ldots, B
\end{equation}
and $s$ is a threshold such that $s = \infty$ if
$\lim_{s \to \infty}\EE_{\mathcal{R}(s)}\left[\frac{1}{B}\sum_{b=1}^B
\wp_b\right] \leq P$, and
\beq \label{eq:r_set}
\mathcal{R}(s) \triangleq\left\{ \hvv \in \mathbb{R}_+^B:
\frac{1}{B}\sum_{b=1}^B \wp_b \leq s \right\},
\eeq
otherwise, $s$ is chosen such that $P=
\EE_{\mathcal{R}(s)}\left[\frac{1}{B}\sum_{b=1}^B
\wp_b\right]$. In~\eqref{eq:minpow_sol}, $\eta$ is now chosen to
satisfy the rate constraint
\begin{equation}
\frac{1}{B}\sum_{b=1}^B I^{\rm awgn}\left(\mmse^{-1}\left(\min\left\{ \frac{Q-1}{Q}, \frac{1}{\eta h^2_b}\right\}\right)\right)=R
\end{equation}

From~\cite{NguyenGuillenRasmussen2007}, the long-term SNR exponent
is given by
\beq
d_{(\log\snr)}^{\rm lt} 
= \begin{cases}
 \frac{d^{\rm st}_{(\log\snr)} }{ 1- d^{\rm st}_{(\log\snr)} } & d^{\rm st}_{(\log\snr)} < 1  \\
 \infty                               & d^{\rm st}_{(\log\snr)} > 1
\end{cases},
\label{eq:lt_exponent}
\eeq where $d^{\rm st}_{(\log\snr)}$ is the short-term SNR exponent,
i.e., the SNR
exponents~\eqref{eq:mimo_snr_exponent_lognormal}-\eqref{eq:mimo_snr_exponent_gamma_gamma}. Note
that $d_{(\log\snr)}^{\rm lt} = \infty $ implies the outage
probability curve is vertical, i.e. the power allocation
scheme~\eqref{eq:long_sol} is able to maintain constant instantaneous
mutual information~\eqref{eq:instant_mi}. The maximum achievable rate
at which this occurs is defined as the~\textit{delay-limited
capacity}~\cite{HanlyTse1998}. From~\eqref{eq:lt_exponent}
and~\eqref{eq:mimo_snr_exponent_lognormal}-\eqref{eq:mimo_snr_exponent_gamma_gamma},
we therefore have the following corollary.

\begin{corollary}\label{cor:dl_cap_req}
The delay-limited capacity of the channel described
by~\eqref{eq:mimo_egc_model2} with CSIT subject to long-term power
constraint~\eqref{eq:lt_pc} is zero whenever \beq MN  \leq
\begin{cases}
2  \left(1+\left\lfloor B\left(1- R_c\right)\right\rfloor \right)^{-1} & \text{exponential} \\
\frac{2}{\min(\alpha, \beta)}  \left(1+\left\lfloor B\left(1- R_c\right) \right\rfloor \right)^{-1} & \text{gamma-gamma}
\end{cases}.
\label{eq:dl_cap_req}
\eeq
For lognormal scintillation, delay-limited capacity is always nonzero.
\end{corollary}
Corollary~\ref{cor:dl_cap_req} outlines fundamental design criteria
for nonzero delay-limited capacity in FSO communications. Single block
transmission ($B=1$) is of particular importance given the slow
time-vary nature of scintillation. From~\eqref{eq:dl_cap_req}, to
obtain nonzero delay-limited capacity with $B=1$, one requires $MN >
2$ and $MN > 2/\min(\alpha,\beta)$ for exponential and gamma-gamma
cases respectively. Note that typically, $\alpha,
\beta \geq 1$. Thus a $3 \times 1$, $1 \times 3$ or $2 \times 2$ MIMO
system is sufficient for most cases of interest.

In addition, for the special case $B=1$, the solution~\eqref{eq:minpow_sol} can be determined explicitly since 
\beq
\eta = \left(h^2 \mmse( I^{\rm awgn, -1}(R)) \right)^{-1} =
\left(h^2 \mmse(\snr_R^{\rm awgn}) \right)^{-1}.
\eeq
 Therefore,
\beq \label{eq:opt_pow_sol_b_1}
\wp^{\rm opt} = \frac{\snr_R^{\rm awgn}}{h^2}.
\eeq
Intuitively,~\eqref{eq:opt_pow_sol_b_1} implies that for single block
transmission, whenever $\snr_R^{\rm awgn}/ h^2 \leq s$, one simply
transmits at the minimum power necessary so that the received
instantaneous SNR is equal to the SNR threshold ($\snr_R^{\rm awgn}$)
of the code. Otherwise, transmission is turned off. Thus an outage occurs
whenever $h < \sqrt{\frac{\snr_R^{\rm awgn}}{s}}$ and
hence  
\beq
P_{\rm out}(\snr,R) = F_H\left(
\sqrt{\frac{\snr^{\rm awgn}_R }{\gamma^{-1} (\snr) }} \right)
\label{eq:csit_pout}
\eeq
where
$\gamma^{-1} (\snr)$ is the solution to the equation $\gamma(s) =
\snr$, i.e.,
\beq
\gamma(s) \triangleq \snr_{R}^{\mathrm{awgn}} \int_{\nu }^{\infty}
\frac{ f_{H}(h)}{h^2} \, dh,
\label{eq:csit_eq}
\eeq
where $\nu \triangleq
\sqrt{\frac{\snr_{R}^{\mathrm{awgn}}}{s}}$. Moreover, the $\snr$ at
which $P_{\rm out}(R,\snr) \rightarrow 0$ is precisely
$\lim_{s\rightarrow \infty} \gamma(s)$. In other words, the minimum long-term
average SNR required to maintain a constant mutual information of $R$ bits per
channel use, denoted by $\snrbar$, is
\beq
\overline{\mathsf{snr}}^{\rm \, awgn}_{R} = \snr^{\rm awgn}_{R} \int_0^{\infty} \frac{f_{H}(h)}{h^2} \, dh. \label{eq:lt_ave_snr_thresh}
\eeq

Hence, recalling that $\snr_R^{\rm awgn}=I^{\rm awgn, -1}(R)$, the
delay-limited capacity (under the constraint of PPM) is\footnote{Note
that a similar expression was derived in~\cite{HanlyTse1998}.}
\beq \label{eq:ppm_dl_cap}
C_{d}(\snr) = I^{\rm awgn} \left( \frac{\snr}{\int_{0}^{\infty} \frac{f_H(h)}{h^2} \,dh  } \right).
\eeq

In the cases where the distribution of $H$ is known in closed form,
\eqref{eq:csit_eq} can be solved explicitly, hence yielding the exact
expressions for outage probability~\eqref{eq:csit_pout} and
delay-limited capacity~\eqref{eq:ppm_dl_cap}. For lognormal
distributed scintillation with $B=M=N=1$, we have that
\beq
\gamma^{\rm{ln}} (s) = \frac{1}{2} \snr_R^{\mathrm{awgn}} (1 + \sigma^2_I)^4 \, \erfc{ \frac{3 \log(1 + \sigma^2_I) + \frac{1}{2} \log \snr^{\rm awgn}_{R} - \frac{1}{2} \log s  }{ \sqrt{2 \log(1 + \sigma^2_I)}}  },
\eeq
and
\begin{equation}
C^{\ln}_{d}(\snr) = I^{\rm awgn} \left( \frac{\snr}{(1 + \sigma^2_I)^{4}} \right),
\end{equation}
where we have explicitly solved the integrals in~\eqref{eq:csit_eq} and~\eqref{eq:ppm_dl_cap} respectively.

For the exponential case with $B=1$, we obtain,
\beq
\gamma^{\mathrm{exp}}(s) = \snr^{\mathrm{awgn}}_{R}\frac{MN(1 + MN) }{(MN-1)(MN-2)} \bar\Gamma\left(MN-2, \sqrt{MN(1+MN) \frac{\snr^{\mathrm{awgn}}_{R}}{s}} \right),
\eeq
and
\begin{equation}
 C_d^{\exp}(\snr) = 
\begin{cases}
 I^{\rm awgn} \left( \frac{(MN-1)(MN-2)}{MN(1+MN)} \snr \right) & MN > 2 \\
0 & \text{otherwise.}
\end{cases}
\end{equation}
For the gamma-gamma case with $B=M=N=1$, $\gamma^{\mathrm{gg}}(s)$ can
be expressed in terms of hypergeometric functions, which are
omitted for space reasons. The delay-limited capacity, however,
reduces to a simpler expression\footnote{Note that since we assume the normalisation $\EE[H^2] = 1$, then $\int_{0}^{\infty} \frac{f_{H}(h)}{h^2} \, dh = \frac{1}{c^2} \int_{0}^{\infty} \frac{f^{\rm gg}_{\tilde{H}}(u)}{u^2} \, du $, where $c = 1/\sqrt{1 + \sigma^2_I}$ and $f^{\rm gg}_{\tilde{H}}(h)$ is defined as in~\eqref{eq:gg_dist} such that $\EE[\tilde{H}] = 1$.}
\beq
C_d^{\rm gg}(\snr) = 
\begin{cases}
 I^{\rm awgn} \left( \frac{(\alpha-2)(\alpha-1)(\beta-2)(\beta-1)}{(\alpha \beta) (\alpha + 1)(\beta + 1)} \snr \right) & \alpha, \beta > 2 \\
0 & \text{otherwise.}
\end{cases}
\eeq

Fig.~\ref{fig:outage_plot} (right) compares the outage probability for
the $B=1$ CSIT case (with long-term power constraints) for each of the
scintillation distributions. For $MN = 1$ we see that the outage curve
is vertical only for the lognormal case, since $C_d = 0$ for the
exponential and gamma-gamma cases. In these cases one must code over
multiple blocks for $C_d > 0$, i.e. from
Corollary~\ref{cor:dl_cap_req}, $B \geq 6$ and $B \geq 4$ for the
exponential and gamma-gamma cases respectively (with $R_c =
1/2$). Comparing the CSIR and CSIT cases in Fig.~\ref{fig:outage_plot}
we can see that very large power savings are possible when CSI is
known at the transmitter. These savings are further illustrated in
Table~\ref{table:csir_csit_comp}, which compares the SNR required to
achieve $P_{\rm{out}} < 10^{-5}$ (denoted by $\snr^*$) for the CSIR
case, and the long-term average SNR required for $P_{\rm out}
\rightarrow 0$ in the CSIT case (denoted by $\snrbar$, which is given
by~\eqref{eq:lt_ave_snr_thresh}). Note that in the CSIT case, the
values of $\snrbar$ given in the parentheses' is the minimum SNR
required to achieve $P_{\rm out} < 10^{-5}$, since $C_d = 0$ for these
cases (i.e. $\snrbar = \infty$). From Table~\ref{table:csir_csit_comp}
we see that the power saving is at least around 15 dB, and in some
cases as high as 50 dB. We also see the combined benefits of MIMO and
power control, e.g. at $MN = 4$, the system is only 3.7 dB (lognormal)
to 5.2 dB (exponential) from the capacity of nonfading PPM channel
($\snr^{\rm awgn}_{1/2} = 3.18$ dB).

\begin{table}[t]
\caption{Comparison of CSIR and CSIT cases with $B=1$, $R=1/2$, $Q = 2$ $\sigma_I^2=1$, $\alpha = 2$, $\beta = 3$. Both $\snr^*$ and $\snrbar$ are measured in decibels.}
\centering
\begin{tabular}{@{}ccccccc@{}}
\toprule
      & \multicolumn{2}{c}{ lognormal} &  \multicolumn{2}{c}{ exponential} &  \multicolumn{2}{c}{ gamma-gamma} \\ 
 $MN$ & $\snr^*$ & $\snrbar$ & $\snr^*$ & $\snrbar$ & $\snr^*$ & $\snrbar$ \\ \midrule
1     & 40.1    & 15.2  & 106.2 & (56.2)   & 65.6  & (24.5)    \\
2     & 29.2    & 9.9   & 57.9  & (17.8)   & 40.7  & 12.2        \\
3     & 24.4    & 7.9   & 42.0  & 11.0     & 31.7  & 9.0         \\
4     & 21.5    & 6.9   & 34.1  & 8.4      & 26.9  & 7.5         \\
\bottomrule
\end{tabular}
\label{table:csir_csit_comp}
\end{table}

\section{Conclusion} \label{sec:conc}
In this paper we have analysed the outage probability of the MIMO
Gaussian FSO channel under the assumption of PPM and non-ideal
photodetection, for lognormal, exponential and gamma-gamma distributed
scintillation. When CSI is known only at the receiver, we have shown
that the SNR exponent is proportional to the number lasers and
apertures, times a channel related parameter (dependent on the
scintillation distribution), times the Singleton bound, even in the
cases where a closed form expression of the equivalent SISO channel
distribution is not available in closed-form. When the scintillation
is lognormal distributed, we have shown that the outage probability is
dominated by a $(\log(\snr))^2$ term, whereas for the exponential and
gamma-gamma cases it is dominated by a $\log
(\snr)$ term.  When CSI is also known at the transmitter, we applied
the power control techniques of \cite{NguyenGuillenRasmussen2007} to
PPM to show that very significant power savings are possible.

\appendices



\section{Proof of Theorem \ref{theorem:mmse}}
\label{appendix:proof_mmse}
Suppose PPM symbols are transmitted over an AWGN channel, the
non-fading equivalent of~\eqref{eq:mimo_egc_model2}. The received noisy
symbols are given by $\yvv = \sqrt{\rho} \xvv + \zvv$, where $\xvv \in
\Xc^{\rm ppm}$ (we have dropped the time index $\ell$ for brevity of
notation).

Using Bayes' rule~\cite{pap91}, the MMSE estimate is
\beq
\hat\xvv = \EE \left[ \xvv | \yvv \right] = \sum_{q = 1}^{Q} \frac{\evv_q \exp( \sqrt{\rho} y_q)   }{ \sum_{k = 1}^{Q} \exp( \sqrt{\rho} y_k) }. \label{eq:mmse_est7}
\eeq
From~\eqref{eq:mmse_est7} the $i$th element of $\hat\xvv$ is
\begin{equation}
\hat{x}_i = \frac{\exp( \sqrt{\rho} y_i)   }{ \sum_{k = 1}^{Q} \exp( \sqrt{\rho} y_k)  }. \label{eq:x_hat_elem}
\end{equation}
Using the orthogonality principle~\cite{kay93} %
$
\mmse(\rho) = \EE \left[ \|\xvv - \hat\xvv \|^2 \right]
            = \EE[\|\xvv \|^2] - \E[ \| \hat\xvv \|^2]
$. %
Since $\| \evv_q \|^2 = 1 $ for all $q = 1,\ldots, Q$, then
$\EE[\|\xvv \|^2] = 1$. Due to the symmetry of $Q$PPM we need only consider the case when $\xvv= \evv_1$ was transmitted. Hence, %
\beq \label{eq:mmse_work}
\mmse(\rho) = 1 - \left(\EE[\hat{x}_1^2] + (Q-1) \EE[\hat{x}_2^2 ] \right).
\eeq %
Now $y_1 = \sqrt{\rho} + z_1$ and $y_i = z_i$ for $i = 2,\ldots, Q$,
where $z_q$ is a realisation of a random variable $Z_q \sim \Nc(0,1)$
for $q = 1, \ldots, Q$. Hence, substituting these values
in~\eqref{eq:x_hat_elem} and taking the
expectation~\eqref{eq:mmse_work} yields the result given the theorem.

\section{Proof of Theorem~\ref{theorem:mimo_snr_exponents}}
\label{appendix:proof_csir_snr_exponents}
We begin by defining a normalised (with respect to SNR) fading
coefficient, $ \zeta_b^{m,n}=-\frac{2\log\tilde{h}_{b}^{m,n}
}{\log\snr}$, which has a pdf
\beq \label{eq:general_alpha_pdf}
f_{\zeta_b^{m,n}}(\zeta) = \frac{\log\snr}{2} e^{-\frac{1}{2}\zeta \log \snr } \, f_{\tilde{H}} \left( e^{-\frac{1}{2}\zeta \log \snr }\right).
\eeq
Since we are only concerned with the asymptotic outage behaviour, the
scaling of the coefficients is irrelevant, and to simplify our analysis we
assume $\EE[\hat{H}^2] = 1$. Hence the instantaneous SNR for block $b$
is given by
\beq
\rho_b = \snr h_b^2 = \left(\frac{1}{MN}\sum_{m=1}^M\sum_{n=1}^N \snr^{\frac{1}{2}\left(1-\zeta_{b}^{m,n}\right)}\right)^2
\eeq
for $b = 1,\ldots, B$. Therefore,
\begin{align}
\lim_{\snr\to\infty} I^{\rm awgn}(\rho_b) &=
\begin{cases}
0 &\text{if all $\zeta^{m,n}_b>1$}\\
\log_2 Q &\text{at least one $\zeta^{m,n}_b<1$}
\end{cases} \notag \\
&=\log_2 Q\left( 1- \openone\{\zetav_b\succ\onevv\}\right) \notag
\end{align}
where $\zetav_b\eqdef(\zeta^{1,1}_b,\dotsc,\zeta^{M,N}_b)$, $\openone \{ \cdot \}$ denotes the indicator function, $\onevv \eqdef (1,\ldots,1)$ is a $1 \times MN$ vector of 1's, and the notation $\avv \succ \bvv$ for vectors $\avv,\bvv\in\RR^k$ means that $a_i>b_i$ for $i=1,\dotsc,k$.

From the definition of outage probability~\eqref{eq:pout_def}, we have
\begin{align}
P_{\rm out}(\snr,R) = \Pr(I_\hvv(\snr)<R) =\int_\Ac f(\zetav)d\zetav
\end{align}
where $\zetav \eqdef ( \zetav_1, \ldots, \zetav_B)$ is a $1 \times BMN$ vector of normalised fading coefficients, $f(\zetav)$ denotes their joint pdf, and
\begin{align}
\Ac = \left\{\zetav\in\RR^{BMN} :  \sum_{b=1}^B\openone\{\zetav_b\succ\onevv\}>B\left(1 -R_c\right)\right\}
\end{align}
is the asymptotic outage set. We now compute the asymptotic behaviour of the
outage probability, i.e.
\beq \label{eq:lim_neg_log_pout}
-\lim_{\snr\to\infty} \log P_{\rm out}(\snr,R) = -\lim_{\snr\to\infty} \log \int_\Ac f(\zetav)d\zetav.
\eeq
\subsection{Lognormal case} \label{appendix:proof_csir_snr_exponents:lognormal}
From~\eqref{eq:lognormal_pdf} and~\eqref{eq:general_alpha_pdf} we obtain the joint pdf,
\begin{align}
f(\zetav) \doteq  \exp \left(-\frac{(\log\snr)^2}{8\sigma^2}\sum_{b=1}^B \sum_{m=1}^M\sum_{n=1}^N (\zeta^{m,n}_b)^2\right), \label{eq:lognorm_joint_zeta}
\end{align}
where we have ignored terms of order less than $(\log \snr )^2$ in the
exponent and constant terms independent of $\zetav$ in front of the
exponential. Combining~\eqref{eq:lim_neg_log_pout},~\eqref{eq:lognorm_joint_zeta},
and using~Varadhan's~lemma~\cite{dembo_zeitouni},
\beq
-\lim_{\snr\to\infty} \log P_{\rm out}(\snr,R) = \frac{(\log\snr)^2}{8\sigma^2}\inf_\Ac \left\{\sum_{b=1}^B \sum_{m=1}^M \sum_{n=1}^N (\zeta^{m,n}_b)^2\right\} \notag
\eeq
The above infimum occurs when any $\kappa$ of the $\zetav_b$ vectors are
such that $\zetav_b \succ \onevv$ and the other $B - \kappa$ vectors
are zero, where $\kappa$ is a unique integer satisfying
\beq \label{eq:kappa_def}
\kappa < B\left(1 -R_c\right) \leq \kappa+1.
\eeq
Hence, it follows that $\kappa = 1  + \left\lfloor B\left(1-R_c\right) \right\rfloor$ and thus,
\beq
-\lim_{\snr\to\infty} \log P_{\rm out}(\snr,R) =  \frac{(\log\snr)^2}{8\sigma^2} MN  \left (  1  + \left\lfloor B\left(1-R_c\right) \right\rfloor \right). \label{eq:lognorm_logpout}
\eeq
Dividing both sides of~\eqref{eq:lognorm_logpout} by $(\log \snr)^2$
the SNR exponent~\eqref{eq:mimo_snr_exponent_lognormal} is obtained.

\subsection{Exponential case}
\label{appendix:proof_csir_snr_exponents:exponential}
From~\eqref{eq:exp_scint} and~\eqref{eq:general_alpha_pdf} we obtain the joint pdf,
\begin{align}
f ( \zetav ) \doteq \exp \left( - \log \snr \frac{MN}{2} \sum_{b=1}^{B} \sum_{m=1}^{M} \sum_{n=1}^{N} \zeta_b^{m,n} \right),
\end{align}
where we have ignored exponential terms in the exponent and constant
terms independent of $\zetav$ in front of the exponential.

Following the same steps as the lognormal case i.e. the defining the
same asymptotic outage set and application of Varadhan's
lemma~\cite{dembo_zeitouni}, the SNR exponent~\eqref{eq:mimo_snr_exponent_exponential} is obtained.

\subsection{Gamma-gamma case}
\label{appendix:proof_csir_snr_exponents:gamma_gamma}
Let us first assume $\alpha > \beta$. From~\eqref{eq:gg_dist} and~\eqref{eq:general_alpha_pdf} we obtain the joint pdf,
\begin{align}
f_{\zeta_b^{m,n}} (\zeta) \doteq \exp\left( - \frac{\beta}{2} \zeta \log \snr \right), ~~~ \zeta>0\label{eq:f_alpha_gamma_gamma_temp}
 \end{align}
for large $\snr$, where we have used the approximation $\K_{\nu}(x)
\approx \frac{1}{2}\Gamma(\nu) (\frac{1}{2}x)^{-\nu} $ for small $x$
and $\nu > 0$ \cite[p.~375]{abramowitz_stegun}. The extra condition,
$\zeta >0 $, is required to ensure the argument of the Bessel function
approaches zero as $\snr \rightarrow \infty$ to satisfy the
requirements of the aforementioned approximation. For the case $\beta
> \alpha$ we need only swap $\alpha$ and $\beta$
in~\eqref{eq:f_alpha_gamma_gamma_temp}. Hence we have the joint pdf
\beq
f(\zetav) \doteq  \exp\left( - \frac{\min(\alpha,\beta) \log \snr}{2} \sum_{b=1}^{B}  \sum_{m=1}^{M} \sum_{n=1}^{N}\zeta_b^{m,n} \right),~~~\zetav\succ\zerovv. \label{eq:f_alphav_gamma_gamma}
\eeq
Now, following the same steps as in the lognormal and exponential
cases, with the additional constraint $\zetav_b\succ\zerovv$, the SNR
exponent~\eqref{eq:mimo_snr_exponent_gamma_gamma} is obtained




\begin{thebibliography}{10}

\bibitem{wille02}
H.~Willebrand and B.~S. Ghuman,
\newblock {\em Free-Space Optics: Enabling Optical Connectivity in Today's
  Networks},
\newblock Sams Publishing, Indianapolis, USA, 2002.

\bibitem{stro78}
J.~W. Strohbehn, Ed.,
\newblock {\em Laser Beam Propagation in the Atmosphere}, vol.~25,
\newblock Springer-Verlag, Germany, 1978.

\bibitem{gali95}
R.~M. Gagliardi and S.~Karp,
\newblock {\em Optical communications},
\newblock John Wiley \& Sons, Inc., Canada, 1995.

\bibitem{andr05}
L.~C. Andrews and R.~L. Phillips,
\newblock {\em Laser Beam Propagation through Random Media},
\newblock SPIE Press, USA, 2nd edition, 2005.

\bibitem{LetzepisHollandCowley2008}
N.~Letzepis, I.~Holland, and W.~Cowley,
\newblock ``The {G}aussian free space optical {MIMO} channel with {$Q$}-ary
  pulse position modulation,''
\newblock {\em to appear IEEE Trans. Wireless Commun.}, 2008.

\bibitem{ChakDeyFran07}
K.~Chakraborty, S.~Dey, and M.~Franceschetti,
\newblock ``On outage capacity of {MIMO} {P}oisson fading channels,''
\newblock in {\em Proc. IEEE Int. Symp. Inform. Theory}, July 2007.

\bibitem{ChakPrak07}
K.~Chakraborty and P.~Narayan,
\newblock ``The {P}oisson fading channel,''
\newblock {\em IEEE Trans. Inform. Theory}, vol. 53, no. 7, pp. 2349--2364,
  July 2007.

\bibitem{cvi07}
N.~Cvijetic, S.~G. Wilson, and M.~Brandt-Pearce,
\newblock ``Receiver optimization in turbulent free-space optical {MIMO}
  channels with {APD}s and {$Q$}-ary {PPM},''
\newblock {\em IEEE Photon. Tech. Let.}, vol. 19, no. 2, pp. 103--105, Jan.
  2007.

\bibitem{djor06}
I.~B. Djordjevic, B.~Vasic, and M.~A. Neifeld,
\newblock ``Multilevel coding in free-space optical {MIMO} transmission with
  q-ary {PPM} over the atmospheric turbulence channel,''
\newblock {\em IEEE Photon. Tech. Let.}, vol. 18, no. 14, pp. 1491--1493, July
  2006.

\bibitem{wilson05aug}
S.~G. Wilson, M.~Brandt-Pearce, Q.~Cao, and J.~H. Leveque,
\newblock ``Free-space optical {MIMO} transmission with {$Q$}-ary {PPM},''
\newblock {\em IEEE Trans. on Commun.}, vol. 53, no. 8, pp. 1402--1412, Aug.
  2005.

\bibitem{chak05}
K.~Chakraborty,
\newblock ``Capacity of the {MIMO} optical fading channel,''
\newblock in {\em Proc. IEEE Int. Symp. Inform. Theory}, Adelaide, Sept. 2005,
  pp. 530--534.

\bibitem{lee04}
E.~J. Lee and V.~W.~S. Chan,
\newblock ``Part 1: optical communication over the clear turbulent atmospheric
  channel using diversity,''
\newblock {\em J. Select. Areas Commun.}, vol. 22, no. 9, pp. 1896--1906, Nov.
  2005.

\bibitem{haas03}
S.~M. Haas and J.~H. Shapiro,
\newblock ``Capacity of wireless optical communications,''
\newblock {\em IEEE J. Select. Areas Commun.}, vol. 21, no. 8, pp. 1346--1356,
  Oct. 2003.

\bibitem{far07}
A.~A. Farid and S.~Hranilovic,
\newblock ``Outage capacity optimization for free-space optical links with
  pointing errors,''
\newblock {\em IEEE Trans. Light. Tech.}, vol. 25, no. 7, pp. 1702--1710, July
  2007.

\bibitem{ozarow_shamai_wyner}
{L. H. Ozarow, S. Shamai and A. D. Wyner},
\newblock ``Information theoretic considerations for cellular mobile radio,''
\newblock {\em IEEE Trans. on Vehicular Tech.}, vol. 43, no. 2, pp. 359--378,
  May 1994.

\bibitem{biglieri_proakis_shamai}
{E. Biglieri, J. Proakis and S. Shamai},
\newblock ``Fading channels: information-theoretic and communications
  aspects,''
\newblock {\em IEEE Trans. on Inform. Theory}, vol. 44, no. 6, pp. 2619 --2692,
  Oct. 1998.

\bibitem{haba01}
M.~A. Al-Habash, L.~C. Andrews, and R.~L. Phillips,
\newblock ``Mathematical model for the irradiance probability density function
  of a laser beam propagating through turbulent media,''
\newblock {\em SPIE Opt. Eng.}, vol. 40, no. 8, pp. 1554--1562, 2001.

\bibitem{KnoppHumblet}
R.~Knopp and P.~Humblet,
\newblock ``On coding for block fading channels,''
\newblock {\em IEEE Trans. on Inform. Theory}, vol. 46, no. 1, pp. 1643--1646,
  July 1999.

\bibitem{malkamaki_leib}
{E. Malkamaki and H. Leib},
\newblock ``Coded diversity on block-fading channels,''
\newblock {\em IEEE Trans. on Inform. Theory}, vol. 45, no. 2, pp. 771--781,
  March 1999.

\bibitem{albert_beppe_it}
{A. Guill\'en i F\`abregas and G. Caire},
\newblock ``Coded modulation in the block-fading channel: Coding theorems and
  code construction,''
\newblock {\em IEEE Trans. on Information Theory}, vol. 52, no. 1, pp.
  262--271, Jan. 2006.

\bibitem{caire_taricco_biglieri_pc}
{G. Caire, G. Taricco and E. Biglieri},
\newblock ``Optimum power control over fading channels,''
\newblock {\em IEEE Trans. on Inform. Theory}, vol. 45, no. 5, pp. 1468--1489,
  July 1999.

\bibitem{NguyenGuillenRasmussen2007}
K.~D. Nguyen, A.~Guill\'en i~F\`abregas, and L.~K. Rasmussen,
\newblock ``Power allocation for discrete-input delay-limied fading channels,''
\newblock {\em submitted to IEEE Trans. Inf. Theory., {\tt
  http://arxiv.org/abs/0706.2033}}, Jun. 2007.

\bibitem{HanlyTse1998}
S.~V. Hanly and D.~N.~C. Tse,
\newblock ``Multiaccess fading channels. {II}. delay-limited capacities,''
\newblock {\em IEEE Trans. on Inform. Theory}, vol. 44, no. 7, pp. 2816--2831,
  Nov. 1998.

\bibitem{doli00b}
S.~Dolinar, D.~Divsalar, J.~Hamkins, and F.~Pollara,
\newblock ``Capacity of pulse-position modulation ({PPM}) on {G}aussian and
  {W}ebb channels,''
\newblock {\em {JPL} {TMO} Progress Report 42-142}, Aug. 2000,
\newblock {URL}: lasers.jpl.nasa.gov/PAPERS/OSA/142h.pdf.

\bibitem{doli00}
S.~Dolinar, D.~Divsalar, J.~Hamkins, and F.~Pollara,
\newblock ``Capacity of {PPM} on {APD}-detected optical channels,''
\newblock in {\em 21st Cent. Millitary Commun. Conf. Proc.}, Oct. 2000, vol.~2,
  pp. 876--880.

\bibitem{zhu03}
X.~Zhu and J.~M. Kahn,
\newblock ``Performance bounds for coded free-space optical communications
  through atmospheric turbulence channels,''
\newblock {\em IEEE Trans. on Commun.}, vol. 51, no. 8, pp. 1233--1239, Aug.
  2003.

\bibitem{li03}
J.~Li and M.~Uyasl,
\newblock ``Optical wireless communications: system model, capacity and
  coding,''
\newblock in {\em IEEE 58th Vehicular Tech. Conf.}, Oct. 2003, vol.~1, pp.
  168--172.

\bibitem{simon05}
M.~K. Simon and V.~A. Vilnrotter,
\newblock ``Alamouti-type space-time coding for free-space optical
  communication with direct detection,''
\newblock {\em IEEE Trans. Wireless Commun.}, , no. 1, pp. 35--39, Jan 2005.

\bibitem{navi07}
S.~M. Navidpour, M.~Uysal, and M.~Kavehrad,
\newblock ``{BER} performance of free-space optical transmission with spatial
  diversity,''
\newblock {\em IEEE Trans. on Wireless Commun.}, vol. 6, no. 8, pp. 2813--2819,
  August 2007.

\bibitem{andr99}
L.~C. Andrews, R.~L. Phillips, C.~Y. Hopen, and M.~A. Al-Habash,
\newblock ``Theory of optical scintillation,''
\newblock {\em J. Opt. Soc. Am. A}, vol. 16, no. 6, pp. 1417--1429, June 1999.

\bibitem{lawr70}
R.~S. Lawrence and J.~W. Strohbehn,
\newblock ``A survey of clean-air propagation effects relevant to optical
  communications,''
\newblock {\em Proc. IEEE}, vol. 58, no. 10, pp. 1523--1545, Oct. 1970.

\bibitem{chur87}
J.~H. Churnside and S.~F. Clifford,
\newblock ``Log-normal {R}ician probability-density function of optical
  scintillations in the turbulent atmosphere,''
\newblock {\em J. Opt. Soc. Am. A}, vol. 4, no. 10, pp. 1923--1930, Oct. 1987.

\bibitem{hill97}
R.~J. Hill and R.~G. Frehlich,
\newblock ``Probability distribution of irradiance for the onset of strong
  scintillation,''
\newblock {\em J. Opt. Soc. Am. A}, vol. 14, no. 7, pp. 1530--1540, July 1997.

\bibitem{Vet07}
F.~S. Vetelino, C.~Young, L.~Andrews, and J.~Recolons,
\newblock ``Aperture averaging effects on the probability density of
  irrandiance fluctuations in moderate-to-strong turbulence,''
\newblock {\em Applied Optics}, vol. 46, no. 11, pp. 2099--2108, April 2007.

\bibitem{beck67}
P.~Beckmann,
\newblock {\em Probability in Communication Engineering},
\newblock Harcourt, Brace and World, New York, 1967.

\bibitem{abramowitz_stegun}
M.~Abramowitz and I.~A. Stegun,
\newblock {\em Handbook of Mathematical Functions with Formulas, Graphs and
  Mathematical Tables},
\newblock New York: Dover Press, 1972.

\bibitem{verdu_han}
{S. Verd\'u and T. S. Han},
\newblock ``A general formula for channel capacity,''
\newblock {\em IEEE Trans. on Inform. Theory}, vol. 40, no. 4, pp. 1147--1157,
  Jul. 1994.

\bibitem{cover_thomas}
T.~M. Cover and J.~A. Thomas,
\newblock {\em Elements of Information Theory},
\newblock Wiley Series in Telecommunications, 1991.

\bibitem{GuoShamaiVerdu2005}
D.~Guo, S.~Shamai, and S.~Verd\'{u},
\newblock ``Mutual information and minimum mean-square error in {G}aussian
  channels,''
\newblock {\em IEEE Trans. Inf. Theory}, vol. 51, no. 4, pp. 1261--1282, Apr.
  2005.

\bibitem{Tse}
L.~Zheng and D.~Tse,
\newblock ``Diversity and multiplexing: A fundamental tradeoff in multiple
  antenna channels,''
\newblock {\em IEEE Trans. on Inform. Theory}, vol. 49, no. 5, May 2003.

\bibitem{mitchell68}
R.~L. Mitchell,
\newblock ``Permanence of the log-normal distribution,''
\newblock {\em J. Opt. Soc. Am.}, 1968.

\bibitem{slimane01}
S.~B. Slimane,
\newblock ``Bounds on the distribution of a sum of independent lognormal random
  variables,''
\newblock {\em IEEE Trans. on Commun.}, vol. 49, no. 6, pp. 975--978, June
  2001.

\bibitem{beau04a}
N.~C. Beaulieu and X.~Qiong,
\newblock ``An optimal lognormal approximation to lognormal sum
  distributions,''
\newblock {\em IEEE Trans. Vehic. Tech.}, vol. 53, no. 2, March 2004.

\bibitem{beau04b}
N.~C. Beaulieu and F.~Rajwani,
\newblock ``Highly accurate simple closed-form approximations to lognormal sum
  distributions and densities,''
\newblock {\em IEEE Commun. Let.}, vol. 8, no. 12, pp. 709--711, Dec. 2004.

\bibitem{LozanoTulinoVerdu2006}
A.~Lozano, A.~M. Tulino, and S.~Verd\'u,
\newblock ``Optimum power allocation for parallel {G}aussian channels with
  arbitrary input distributions,''
\newblock {\em IEEE Trans. Inform. Theory}, vol. 52, no. 7, pp. 3033--3051,
  July 2006.

\bibitem{pap91}
A.~Papoulis,
\newblock {\em Probability, random variables, and stochastic processes},
\newblock McGraw-Hill, 1991.

\bibitem{kay93}
S.~M. Kay,
\newblock {\em Fundamentals of Statistical Signal Processing: Estimation
  Theory},
\newblock Prentice Hall Int. Inc., USA, 1993.

\bibitem{dembo_zeitouni}
A.~Dembo and O.~Zeitouni,
\newblock {\em Large Deviations Techniques and Applications},
\newblock Number~38 in Applications of Mathematics. Springer Verlag, 2nd
  edition, April 1998.

\end{thebibliography}
\end{document}